\documentclass{amsart}

\usepackage{color}
\usepackage{latexsym}
\usepackage{amsfonts}

\usepackage[svgnames]{xcolor}
\definecolor{linkcolor}{HTML}{e88d67} 
\definecolor{citecolor}{HTML}{e88d67} 
\definecolor{urlcolor}{HTML}{e88d67} 

\usepackage[
  colorlinks = , 
  urlcolor = urlcolor, 
  linkcolor = linkcolor, 
  citecolor = citecolor, 
  linktoc = all
]{hyperref}

\theoremstyle{plain}
\newtheorem{thm}{Theorem}[]

\newtheorem{lem}{Lemma}[]

\theoremstyle{definition}

\newtheorem{exmp}{Example}[]

\theoremstyle{remark}
\newtheorem{rem}{Remark}

\usepackage{chngcntr}
\counterwithin*{thm}{section} 
\counterwithin*{lem}{section} 
\counterwithin*{claim}{section} 
\counterwithin*{prop}{section} 
\counterwithin*{corl}{section} 
\counterwithin*{defn}{section} 
\counterwithin*{exmp}{section} 
\counterwithin*{ex}{section} 
\counterwithin*{rem}{section} 

\usepackage{apptools}
\AtAppendix{\counterwithin{thm}{section}}
\AtAppendix{\counterwithin{lem}{section}}
\AtAppendix{\counterwithin{claim}{section}}
\AtAppendix{\counterwithin{prop}{section}}
\AtAppendix{\counterwithin{corl}{section}}
\AtAppendix{\counterwithin{defn}{section}}
\AtAppendix{\counterwithin{exmp}{section}}
\AtAppendix{\counterwithin{ex}{section}}
\AtAppendix{\counterwithin{rem}{section}}

\DeclareMathOperator{\const}{const}
\newcommand{\brackets}[1]{\left( #1 \right)}
\newcommand{\PoissonBrackets}[1]{\left\{ #1 \right\}}
\newcommand{\LieBrackets}[1]{\left[ #1 \right]}

\newcommand{\PIIn}[1]{\text{P}_\text{II}^{\brackets{ #1 }}}

\newcommand{\dsum}{\displaystyle \sum}

\newcommand{\Painleve}{Painlev{\'e} }
\newcommand{\Backlund}{B{\"a}cklund }

\newcommand{\undertext}[2]{\underset{ #1 }{\underbrace{ #2 }}}

\def\d{\hbox{\rm d}}

\def\od#1{\frac{\d}{\d#1}}

\def\hide#1{}

\newcommand{\p}{\partial}

\def\p{\partial}

\def\Pt{{\rm P}_{\rm{\scriptstyle II}}}
\def\Ptn{{\rm P}^{(n)}_{\rm{\scriptstyle II}}}
\def\Lr{\mathcal{L}}

\begin{document}

\title[Sigma Form for the PII Hierarchy]
{The Sigma Form of the Second Painlev\'e Hierarchy}

\author{Irina Bobrova}
\noindent\address{\noindent Faculty of Mathematics, National Research Institute "Higher School of Economics", Moscow, Russia}
\email{ia.bobrova94@gmail.com}

\author{Marta Mazzocco}
\noindent\address{\noindent School of Mathematics, University of Birmingham, Birmingham, UK}
\email{m.mazzocco@bham.ac.uk}


\maketitle

\begin{abstract}
In this paper we study the so-called sigma form of the  second Painlev\'e hierarchy. To obtain this form, we use some properties of the Hamiltonian structure of the  second Painlev\'e hierarchy and of the Lenard operator.
\end{abstract}

\tableofcontents

\section{Introduction}

The Painlev\'e differential equations were discovered more than a hundred years ago and since the eighties have appeared in many branches on mathematics and physics, including several of Dubrovin's seminal papers on Frobenius manifolds. 

The reason behind the ubiquitous appearance of these equations is that they are innately linked to the Toda hierarchy. In \cite{DZ1}, Dubrovin and Zhang proved that the tau function of a generic solution to the extended Toda hierarchy is annihilated by some combinations of the Virasoro operators. It is such Virasoro constraints that regulate the correlation functions of many systems in random matrix theory, in string theory and topological field theory.  For example in \cite{DZ}, expressions for the genus $g\geq 1$ total Gromov--Witten potential were obtained via the genus zero quantities derived from the Virasoro constraints. 

The link between Toda--type systems and Painlev\'e equations becomes explicit when the latter are re-formulated in the so called {\it sigma form}  introduced in \cite{JMMS} as the equation satisfied by the logarithmic derivative of the isomonodromic $\tau$--function. An other approach to obtain the sigma form was proposed by Okamoto who developed the Hamiltonian theory of the Painlev\'e differential equations and showed that  all  B\"acklund transformations can be obtained as natural affine Weyl groups actions on the sigma form (\cite{okam1}, \cite{okam2}, \cite{okam3}, \cite{okam4}).

In this paper, we present the sigma form for the second Painlev\'e hierarchy, an infinite sequence of non linear ODEs containing
$$
\Pt :\quad\quad w_{zz} = 2 w^3+z\,w+\alpha_1,
$$
as its simplest equation.The
 $n$-th element  is of order $2n$, and depends on $n$ parameters denoted by ${t}_1,\dots,{t}_{n-1}$ and $\alpha_n$:
\begin{equation}\label{PIIhier}
    \Ptn :\,
    \left(
    \od{z}
    + 2 w
    \right) 
    \Lr_{n} \left[w_{z}-w^{2}\right]
    + \sum_{l=1}^{n-1} {t}_{l} \left( \od{z}
    + 2 w \right) \Lr_{l}\left[w_{z}-w^{2}\right]
    = z w 
    + \alpha_{n},
    \, n\ge 1,
\end{equation}
where $\Lr_{n}$ is the
operator defined by the recursion relation
\begin{equation}\label{i:lenard}
    \od{z}
    {\mathcal L}_{n+1}
    = \left(
    \frac{{\rm d^3}}{{\rm d}z^3}
    + 4 \left( w_z 
    - w^2 \right)
    \od{z}
    + 2 \left( w_z 
    - w^2 \right)_z
    \right){{\mathcal L}}_n,
    \quad 
    {\mathcal L}_0 
    \left[w_z - w^2 \right] 
    = {\textstyle\frac12},
\end{equation}
with boundary condition
\begin{equation}\label{eq:L0}
    \mathcal L_{n}[0]
    :=0,
    \quad 
    \forall n \geq 1.
\end{equation}

The Hamiltonian form of the second Painlev\'e hierarchy was produced in \cite{MM} where the authors gave canonical coordinates $P_1, \dots, P_n, Q_1, \dots, Q_n$ and a Hamiltonian function ${\mathcal H}^{(n)}$ such that $\Ptn$ is equivalent to
\begin{equation}\label{i:hpq}
    \frac{\partial Q_i}{\partial z}
     =\frac{\partial {\mathcal H}^{(n)}}{\partial P_i},
    \qquad
    \frac{\partial P_i}{\partial z}
    = - \frac{\partial {\mathcal H}^{(n)}}{\partial Q_i},
    \quad 
    i = 1, \dots, n.
\end{equation}
In particular $ {\mathcal H}^{(n)}$ is a polynomial in
$P_1,\dots,P_n,Q_1,\dots,Q_n$ and that the Hamiltonian equations satisfy the Painlev\'e property.

The sigma function is by definition the evaluation of the Hamiltonian on solutions, namely
\begin{equation}\label{eq-s-def}
    \sigma_n(z)
    := {\mathcal H}^{(n)}
    \left( 
    P_1(z), \dots, P_n(z), Q_1(z), \dots, Q_n(z) 
    \right).
\end{equation}
Our main result in this paper is the following 

\begin{thm} \label{SFthm}
Consider the Lenard operators $\widehat{\mathcal L}_k$ defined by
\begin{align}\label{i:lenards}
    \od{z}
    \widehat{\mathcal L}_{k+1} \left[\sigma_n'-\frac{t_{n-1}}{2}\right]
    &= \left(
    \frac{{\rm d^3}}{{\rm d}z^3}
    + 2 \left(2\sigma_n'- t_{n-1}\right)
    \od{z}
    + 2 \sigma_n'' 
    \right)
    \widehat{\mathcal L}_k\left[\sigma_n'-\frac{t_{n-1}}{2}\right],
    \\
    \notag 
    \widehat{\mathcal L}_0 \left[\sigma_n'-\frac{t_{n-1}}{2}\right]
    &= {\textstyle\frac12},
    \quad 
    t_0 = - z,
\end{align}
with the boundary condition
\begin{equation}\label{eq:L0s}
    \widehat{\mathcal L}_{k}[0]
    := \left(
    - \frac{ t_{n-1}}{2}
    \right)^k
    \frac{1}{k}
    \left(
    \begin{array}{c} 2k\\ k\\ \end{array}
    \right),
\end{equation}
and define
\begin{equation}\label{eq:def-fn}
    f_n 
    = \dsum_{l = 1}^{n} t_l \widehat{\mathcal{L}}_l \LieBrackets{ \sigma_n^{\prime} \brackets{z} - \dfrac{t_{n - 1}}{2}}, 
    \quad 
    t_n = 1.
\end{equation}
Then, for $ n > 1$, the $n$-th element of the second Painlev\'e hierarchy \eqref{PIIhier} is equivalent to 
\begin{equation}\label{SF1}
    - f_n^{\prime} 
    + \brackets{f_n^{\prime}}^2 
    + \brackets{z - 2 f_n}
    \brackets{
    \sum_{l=1}^n t_l 
    \left( \widehat{\mathcal L}_l^{\prime\prime} - \widehat{\mathcal L}_{l+1} \right)
    + 2 f_1 f_n
    + \sigma_n \brackets{z}
    - \frac{1}{2} t_{n-1} z
    }
    = \alpha_n \brackets{\alpha_n - 1}.
\end{equation}
\end{thm}

\begin{rem}
We note that an other equation of order $2n+1$ for the sigma function of the $n$-the element of the second Painlev\'e hierarchy was produced by Stuart Andrew, a former master student of the second author in integral form \cite{SA}. The formula \eqref{SF1} in this paper is an explicit ODE of order $2 n$ as expected (see Lemma \ref{lm:order} in Section \ref{se:proof}).
\end{rem}


\textbf{Acknowledgements.} 
The authors would like to express their gratitude to Volodya Rubtsov for introducing them to each other. 
The authors are also grateful to Vladimir Poberezhnyi, who initiated I.B. to the \Painleve equations theory and constantly supported her during her scientific work.
The research of I.B. is a part of her PhD program studies
at the Higher School of Economics (HSE).  I. B. would like to thank to Faculty of Mathematics for giving her such opportunity. The research of M.M. is supported by the EPSRC Research Grant $EP/P021913/1$. The research of I.B. was partially supported by the RFBR Grant 18-01-00461 A. We dedicate this paper in lasting memory of Boris A. Dubrovin.

\section{Hamiltonian structure of the second Painlev\'e hierarchy}

In this section we resume some results in \cite{MM} that turn out to be useful in our proof of Theorem \ref{SFthm}.

Let us consider the isomonodromic deformation problem for the $\PIIn{n}$ hierarchy
\begin{align*}
\dfrac{\partial \Psi}{\partial z}
&= \mathcal{B} \Psi
=
	\begin{pmatrix}
    -\lambda & w \\
    w & \lambda
    \end{pmatrix}
    \Psi
,
\\
\dfrac{\partial \Psi}{\partial \lambda}
&= \mathcal{A}^{(n)} \Psi
= \dfrac{1}{\lambda}
\left[
\begin{pmatrix}
- \lambda z & - \alpha_n
\\
- \alpha_n & \lambda z
\end{pmatrix}
+ M^{(n)} 
+ \sum_{l = 1}^{n - 1} t_l M^{(l)}
\right]
,
\\
\brackets{2 k + 1} \dfrac{\partial \Psi}{\partial t_k}
&= 
\left(
M^{(k)}
- 
\begin{pmatrix}
0 & \brackets{\partial_z + 2 w} \mathcal{L}_k \LieBrackets{w^{\prime} - w^2}
\\
\brackets{\partial_z + 2 w} \mathcal{L}_k \LieBrackets{w^{\prime} - w^2} & 0
\end{pmatrix}
\right)
\Psi,
\end{align*}
where the matrix $M^{(l)}$ is defined as
\begin{equation*}
M^{(l)}
=
\begin{pmatrix}
\dsum_{j = 1}^{2 l + 1} A_j^{(l)} \lambda^j
&
\dsum_{j = 1}^{2l} B_j^{(l)} \lambda^j
\\
\dsum_{j = 1}^{2l} C_j^{(l)} \lambda^j
&
- \dsum_{j = 1}^{2 l + 1} A_j^{(l)} \lambda^j
\end{pmatrix},
\end{equation*}
with 
\begin{align}
	\label{entrMn}
    A_{2l+1}^{\brackets{l}} 
    &= 4^l;
    \qquad 
    A_{2k}^{\brackets{l}} 
    = 0, 
    \quad
    k 
    = 0, \dots, l; 
    \\
    \notag
	A_{2k+1}^{(l)} 
	&= \dfrac{4^{k+1}}{2} 
	\left\{ {\mathcal{L}}_{l-k} \left[ w^{\prime} - w^2 \right] 
	- \od{z} \left( \od{z} + 2 w \right) {\mathcal{L}}_{l-k-1} \left[ w^{\prime} - w^2 \right]
	\right\},
    \quad
    k = 0, \dots l -1; 
    \\
    \notag
	B_{2k+1}^{(l)} 
	&= \dfrac{4^{k+1}}{2} \od{z} \left( \od{z} + 2 w \right) 
	{\mathcal{L}}_{l-k-1} \left[ w^{\prime} - w^2 \right],
    \quad 
    k = 0, \dots, l - 1;
    \\
    \notag
	B_{2k}^{(l)} 
	&= - 4^k \left( \od{z} + 2 w \right) 
	{\mathcal{L}}_{l-k} \left[ w^{\prime} - w^2 \right], 
    \quad
    k = 1, \dots, l.
\end{align}

The compatibility condition 
\begin{equation*}
    \dfrac{\partial \mathcal{A}^{(n)}}{\partial z} 
    - \dfrac{\partial \mathcal{B}}{\partial \lambda} 
    = \LieBrackets{\mathcal{B}, \mathcal{A}^{(n)}}
\end{equation*}
gives the $n$-th member of the second \Painleve hierarchy \eqref{PIIhier}.

It is convenient to introduce new notations to define the matrix $\mathcal{A}^{(n)}$. Let us set
\begin{align} \label{entrAn}
	a_{2k+1}^{(n)} 
	&= \sum_{l=1}^n t_l A_{2k+1}^{(l)},
	\quad
	k = 1, \dots, n;
	&	
	a_1^{(n)} 
	&= \sum_{l=1}^{n} t_l A_1^{(l)} - z; 
	\\
	\notag
	b_{2k+1}^{(n)} 
	&= \sum_{l=1}^{n} t_l B_{2k+1}^{(l)},
	\quad
	k = 1, \dots, n - 1;
	\\
	\notag
	b_{2k}^{(n)} 
	&= \sum_{l=1}^{n} t_l B_{2k}^{(l)}, 
	\quad
	k = 1, \dots, n;
	&
	b_0^{(n)} 
	&= - \alpha_n, 
	\quad t_n = 1.
\end{align}

Therefore, $\mathcal{A}^{(n)}$ can be represented in the following form
\begin{equation*}
    \mathcal{A}^{(n)} 
    = 
    \begin{pmatrix}
    \dsum_{k=0}^{n} a_{2k+1}^{(n)} \lambda^{2k} & 
    \dsum_{k=0}^{2n} b_{k}^{(n)} \lambda^{k-1} \\
    \dsum_{k=0}^{2n} \brackets{-1}^k b_{k}^{(n)} \lambda^{k-1} 	& 
    - \dsum_{k=0}^{n} a_{2k+1}^{(n)} \lambda^{2k} \\
    \end{pmatrix}.
\end{equation*}

The canonical coordinates are given by relations
\begin{gather*}
    P_k 
    = \Pi_{2k} 
    = \dfrac{a_{2(n-k)+1}^{(n)} + b_{2(n-k)+1}^{(n)}}{a_{2n+1}^{(n)}},
    \quad
    Q_k 
    = \sum_{j = 1}^n \dfrac{1}{2j} b_{2j}^{(n)} \dfrac{\partial S_{2j}}{\partial \Pi_{2k}},
    \quad 
    k 
    = 1, \dots, n; 
    \\
    S_k = \sum_{j=1}^{2n} q_j^k,
    \quad
    k = 1, \dots, 2n;
    \\
    \Pi_1 
    = q_1 + \dots + q_{2n}, 
    \quad
    \Pi_2 
    = \sum_{1 \leq j \leq 2n} q_j q_k, 
    \quad
    \dots,
    \quad
    \Pi_{2n} 
    = q_1 q_2 \dots q_{2n},
\end{gather*}
where $q_j$ are solutions of the following equation
\begin{equation*}
    \sum_{k=0}^{n-1} \brackets{b_{2k+1}^{(n)} + a_{2k+1}^{(n)}} q_j^{2k} + a_{2n+1}^{(n)} q_j^{2n} = 0,
\end{equation*}
and $p_j = \displaystyle \sum_{k=0}^{n} b_{2k}^{(n)} q_j^{2k-1}$.

The coordinates $P_1, \dots, P_n$, $Q_1, \dots, Q_n$  are canonical with the Poisson structure
\begin{equation*}
	\PoissonBrackets{P_i, P_j}
	= \PoissonBrackets{Q_i, Q_j}
	= 0, 
	\quad
	\PoissonBrackets{P_i, Q_j}
	= \delta_{ij},
	\quad 
	i, j 
	= 1, \dots, n.
\end{equation*}

The corresponding Hamiltonian in terms of the canonical coordinates is 
\begin{multline} \label{HIIn}
	{\mathcal H}^{(n)}
	\brackets{P_1, \dots, P_n, Q_1, \dots, Q_n, z}
	= - \dfrac{1}{4^n}
	\left(
	\dsum_{l = 0}^{n - 1} 
	a_{2 l + 1}^{(n)} a_{2 (n - l) - 1}^{(n)}
	- \dsum_{l = 0}^{n - 1}
	b_{2 l + 1}^{(n)} b_{2 (n - l) - 1}^{(n)}
	\right.
	\\
	\left.
	+ \dsum_{l = 0}^n
	b_{2l}^{(n)} b_{2 (n - l)}^{(n)}
	\right)
	+ \dfrac{Q_n}{4^n},
\end{multline}
where the coefficients of $\mathcal{A}^{(n)}$ can be expressed as polynomials in the canonical coordinates by theorem 6.1 in \cite{MM}.

We conclude this section be reminding a useful formula valid both for the Lenard operators $\mathcal L_n$ and $\widehat{\mathcal L}_n$ \cite{AFord}:
\begin{align} \label{eq:lzz}
    {\mathcal L}_{n+1} 
    =
    \p_z^2{\mathcal L}_{n}
    &+ 3{\mathcal L}_{n}{\mathcal L}_{1}
    + \dsum_{j = 1}^{l - 1}
    \brackets{{\mathcal{L}}_{l - j}  
    \brackets{4 {\mathcal{L}}_1    {\mathcal{L}}_j  
    - {\mathcal{L}}_{j + 1} 
    + 2 {\mathcal{L}}_{j}^{\prime\prime} 
    }
    - {\mathcal{L}}_j^{\prime}  
    {\mathcal{L}}_{l - j}^{\prime}
    },
    \quad
    n > 1.
\end{align}

\section{Proof of main theorem}\label{se:proof}

Firstly, the following correlation between the sigma coordinates and the solution of the $\Ptn$ equation was proved in \cite{SA}:

\begin{lem} \label{sigmaWcorrel} The sigma function $ \sigma_n$ of definition \eqref{eq-s-def} is related to the solution $w_n$ of \eqref{PIIhier} by the following
$$
w_n^{\prime} - w_n^2 = \sigma_n^{\prime} - \dfrac{t_{n - 1}}{2},
$$ 
where $t_0 = - z$.
\end{lem}

\begin{proof}
By definition \eqref{eq-s-def} of the sigma-coordinates, we have
\begin{equation*}
    \sigma_n(z)
    := {\mathcal H}^{(n)}
    \left( 
    P_1(z), \dots, P_n(z), Q_1(z), \dots, Q_n(z) 
    \right).
\end{equation*}
where $P_1, \dots, P_n$, $Q_1, \dots, Q_n$ are canonical coordinates. Its first derivative is
\begin{equation*}
    \sigma_n^{\prime} \brackets{z} 
    = \PoissonBrackets{{\mathcal H}^{(n)}, {\mathcal H}^{(n)}} 
    + \dfrac{\partial {\mathcal H}^{(n)}}{\partial z} 
    = \dfrac{\partial {\mathcal H}^{(n)}}{\partial z}
    .
\end{equation*}
By formula \eqref{HIIn}, the only term in ${\mathcal H}^{(n)}$ that depends explicitly on $z$ is
\begin{equation*}
    - \dfrac{1}{2^{2n - 1}} a_1^{(n)} a_{2n - 1}^{(n)}.
\end{equation*}
Hence, using formulas \eqref{entrMn} and \eqref{entrAn}, $\sigma_n^{\prime} (z)$ is calculated as 
\begin{align*}
    \sigma_n^{\prime} \brackets{z} 
    &= \dfrac{\partial {\mathcal H}^{(n)}}{\partial z}
    = - \dfrac{1}{2^{2n - 1}} \partial_z 
    \brackets{a_1^{(n)} a_{2n - 1}^{(n)}}
    \\
    &= - \dfrac{1}{2^{2n - 1}} \partial_z 
    \brackets{
    \brackets{\dsum_{l = 1}^n t_l A_1^{(l)} - z} 
    \dsum_{l = 1}^n t_l A_{2n - 1}^{(l)}
    }
    \\
    &= \dfrac{1}{2^{2n - 1}}
    \dsum_{l = 1}^n t_l A_{2n - 1}^{(l)}
    = \dfrac{1}{2^{2n - 1}} \brackets{
    t_n A_{2n - 1}^{(n)} 
    + t_{n - 1} A_{2n - 1}^{(n - 1)} 
    }
    \\
    &= \mathcal{L}_1 \LieBrackets{w_n^{\prime} - w_n^2}
    + \dfrac{t_{n - 1}}{2}
    = w_n^{\prime} - w_n^2
    + \dfrac{t_{n - 1}}{2}.
\end{align*}
\end{proof}

Note that thanks to Lemma \ref{sigmaWcorrel}, we can define all Lenard operators in terms of $\sigma_n (z)$ rather than $w (z)$. However, we need to be careful in the integration step to extract $\mathcal L_{n+1} \LieBrackets{w^{\prime} - w^2}$ from formula \eqref{i:lenard}. The fact that the integrand is exact was proved in \cite{Lax}, and the choice of the integration constant depends on the condition we impose on $\mathcal L_{n}[0]$. 

\begin{lem}
For the $n$-th member of the second Pailev\'e hierarchy, $n > 1$, the Lenard operators defined by \eqref{i:lenard}, \eqref{eq:L0} 
coincide with the operators defined recursively by relations \eqref{i:lenards} with boundary conditions \eqref{eq:L0s}, or in other words
$$
    \mathcal L_k\left[w_n^{\prime} - w_n^2\right] 
    = \widehat{\mathcal L}_k \left[\sigma_n'-\frac{t_{n-1}}{2}\right],
$$
where $w_n$ denotes the solution of \eqref{PIIhier}.
\end{lem}

\begin{proof} 
We prove this statement by induction. We know that 
$$
    {\mathcal L}_1 \left[ w_n^{\prime} - w_n^2 \right]
    = w_n^{\prime}- w_n^2.
$$ 
On the other side, using \eqref{i:lenards} we have 
$$
    \od{z}
    \widehat{\mathcal L}_{1} \left[\sigma_n'-\frac{t_{n-1}}{2}\right]
    =  \sigma_n'',
$$ 
so that 
$$
    \widehat{\mathcal L}_{1} \left[\sigma_n'-\frac{t_{n-1}}{2}\right]
    =  \sigma_n'
    + \const,
$$ 
and by imposing the boundary condition \eqref{eq:L0s}, we obtain
$$
    \widehat{\mathcal L}_{1} \left[\sigma_n'-\frac{t_{n-1}}{2}\right]
    =  \sigma_n' 
    - \frac{ t_{n-1}}{ 2},
$$ 
that  due to Lemma \ref{sigmaWcorrel} gives $\mathcal L_1 \left[ w^{\prime} - w^2 \right] = \widehat{\mathcal L}_1 \left[ \sigma_n'-\frac{t_{n-1}}{2} \right]$.

Let us now assume that the statement is true for $k = l$ and prove it for $k= l + 1$. 

Let us call $c_{l + 1}$ the constant term of $\widehat{\mathcal{L}}_{l + 1}  \LieBrackets{ \sigma_n'-\frac{t_{n-1}}{2} }$. By formula \eqref{eq:lzz}, we have that $c_{l + 1}$ is defined as
\begin{align*}
    c_{l + 1}
    &= 3 c_1 c_l
    + \dsum_{j = 1}^{l - 1}
    c_{l - j}
    \brackets{4 c_1 c_j
    - c_{j + 1}
    }
    .
\end{align*}
This discrete equation is solved by $c_k=\left(- \frac{ t_{n-1}}{ 2}\right)^k\frac{1}{k}\left(\begin{array}{c} 2k\\ k\\ \end{array}\right)$ as we wanted to prove.
\end{proof}

\begin{rem}
When $n = 1$, $t_{n-1}=t_0=-z$ is no longer constant. In this there is no change of boundary condition and the operators $\widehat{\mathcal L}_k$ are simply  the standard operators ${\mathcal L}_k$ defined by \eqref{i:lenard}, \eqref{eq:L0}, applied to $\sigma_1'+ z/2$. In this proof of theorem \ref{SFthm} we consider $n\geq 1$ and prove \eqref{SF1} as well as the following (valid for $n=1$)
\begin{equation}\label{SF1n1}
    - f_1^{\prime} + \brackets{f_1^{\prime}}^2 
    + \brackets{z - 2 f_1} 
    \brackets{
    \left(  \widehat{\mathcal L}_1^{\prime\prime}
    - \widehat{\mathcal L}_{2} \right)
    + 2 f_1^2
    + \sigma_1 \brackets{z}
    + \frac{z^2}{4}}
    = \alpha_1 \brackets{\alpha_1 - 1},
\end{equation}
\end{rem}

Now we can proceed to the proof of theorem \ref{SFthm}.

\begin{proof}[Proof of theorem \ref{SFthm}]
Suppose that $z - 2 f_n \neq 0$, i.e. $\alpha_n \notin \frac{1}{2} \mathbb{Z}$. Then $w\brackets{z}$ is expressed from \eqref{PIIhier} as
\begin{equation} \label{wExpr}
    w = \dfrac{f_n^{\prime} - \alpha_n}{z - 2 f_n}.
\end{equation}

From lemma \ref{sigmaWcorrel} and \eqref{wExpr} we obtain
\begin{equation} \label{gnExpr}
    - f_n^{\prime} 
    + \brackets{f_n^{\prime}}^2 
    + \brackets{z - 2 f_n} f_n^{\prime\prime}
    - \undertext{\widehat{\mathcal{L}}_1 \LieBrackets{\sigma_n^{\prime} - \frac{t_{n-1}}{2}}}{\left( \sigma_n^{\prime} - \frac{t_{n-1}}{2} \right)}
    \brackets{z - 2 f_n}^2 
    = \alpha_n \brackets{\alpha_n - 1}.
\end{equation}
This equation involves derivatives of $\sigma_n \brackets{z}$ of order $2n + 1$. Since we are looking for an ODE of order $2n$, we need to remove these. To this aim, 
we differentiate \eqref{gnExpr} obtaining 
\begin{equation}\label{dgnExpr}
    \brackets{z - 2 f_n}
    \brackets{f_n^{\prime\prime\prime} 
    + 4 \widehat{\mathcal{L}}_1 f_n^{\prime}
    + 2 \widehat{\mathcal{L}}_1^{\prime} f_n 
    - 2 \widehat{\mathcal{L}}_1  - z \widehat{\mathcal{L}}_1^{\prime}}
    = 0,
\end{equation}
and use the Lenard recursion relation \eqref{i:lenard} obtaining
\begin{equation} \label{dgnExpr2}
    \brackets{z - 2 f_n}
    \dfrac{d}{dz} \brackets{
    \sum_{l = 1}^n t_l \widehat{\mathcal{L}}_{l + 1}  
    - z\widehat{\mathcal{L}}_1- \sigma_n \brackets{z} + h_n \brackets{z}
    + \const
    }
    = 0,
\end{equation}
where we understand that $\widehat{\mathcal{L}}_k$ is applied to $\sigma_n'-\frac{t_{n-1}}{2}$  and
\begin{equation*}
   h_n(z)=\dfrac{1}{2} z 
    \brackets{t_{n - 1} \brackets{1 - \delta_{n, 1}} - \dfrac{1}{2} z \, \delta_{n, 1}}.
\end{equation*}
By our assumption, $z - 2 f_n \neq 0$. Thus, \eqref{dgnExpr2} becomes
\begin{equation} \label{dgnExpr4}
    \sum_{l = 1}^n t_l \widehat{\mathcal{L}}_{l + 1}  
    - z \widehat{\mathcal{L}}_1 
    - \sigma_n \brackets{z} 
    + h_n \brackets{z}
    = 0,
\end{equation}
where we have absorbed the integration constant in sigma (constant shifts in sigma do not change the dynamics).

So we have two equations, \eqref{gnExpr} and \eqref{dgnExpr4} that both contain derivatives of $\sigma_n \brackets{z}$ of order $2n+1$. We can replace $f_n''$ in 
 \eqref{gnExpr} by $f_n'' -\left( \sum_{l = 1}^n t_l \widehat{\mathcal{L}}_{l + 1}  
    - z \widehat{\mathcal{L}}_1 
    - \sigma_n \brackets{z} 
    + h_n \brackets{z}\right)$ thus obtaining \eqref{SF1}.
\end{proof}

\begin{lem}\label{lm:order} The sigma form \eqref{SF1} is an ODE of order $2n$.
\end{lem}

\proof By using formula \eqref{eq:lzz} we see immediately that \eqref{SF1} is equivalent to 
\begin{multline*} 
    - f_n^{\prime} 
    + \brackets{f_n^{\prime}}^2 
    + \brackets{z - 2 f_n}
    \left(
    \sigma_n \brackets{z}
    - z \frac{t_{n-1}}{2}
    - \brackets{ \sigma_n^{\prime} \brackets{z} 
    - \dfrac{t_{n - 1}}{2}}
    f_n
    \phantom{\dsum_{j = 1}^{l - 1}\widehat{\mathcal{L}}}
    \right.
    \\
    \left.
    - \left( 
    \dsum_{l = 1}^{n} t_l 
    \dsum_{j = 1}^{l - 1}
    \brackets{\widehat{\mathcal{L}}_{l - j}  
    \brackets{4 \widehat{\mathcal{L}}_1  \widehat{\mathcal{L}}_j  
    - \widehat{\mathcal{L}}_{j + 1} 
    + 2 \widehat{\mathcal{L}}_{j}^{\prime\prime} 
    }
    - \widehat{\mathcal{L}}_j^{\prime}  
    \widehat{\mathcal{L}}_{l - j}^{\prime}
    }
    \right)
    \right)
    = \alpha_n \brackets{\alpha_n - 1},
\end{multline*}
which is an ODE of order $2n$.
\endproof

To demonstrate how our theorem \ref{SFthm} works, we give two examples for cases $n = 1$ and $n = 2$.

\begin{exmp}
For $n = 1$ we use \eqref{SF1n1} to obtain
\begin{align*}
    - \brackets{\sigma_1^{\prime\prime} + \dfrac{1}{2}}
    + \brackets{\sigma_1^{\prime\prime} + \dfrac{1}{2}}^2
    - 2 \sigma_1^{\prime} 
    \brackets{\sigma_1 + \dfrac{1}{4} z^2 - \brackets{\sigma_1^{\prime} + \dfrac{z}{2}}^2}
    &= \alpha_1 \brackets{\alpha_1 - 1},
    \\
    \brackets{\sigma_1^{\prime\prime}}^2 
    - 2 \sigma_1 \sigma_1^{\prime} 
    + 2 z \brackets{\sigma_1^{\prime}}^2 
    + 2 \brackets{\sigma_1^{\prime}}^3 
    &= \brackets{\alpha_1 - \dfrac{1}{2}}^2.
\end{align*}
\begin{rem}
If we consider the following map of the Okamoto Hamiltonian in \cite{okam3}
\begin{equation*}
    \text{H}_{\text{II}} \brackets{p, q} 
    = \dfrac{1}{2} p \brackets{p - 2 q^2 - z}
    - \brackets{\alpha + \dfrac{1}{2}} q
    \quad
    \mapsto 
    \quad 
    2 \text{H}_{\text{II}} \brackets{2 p, \frac{1}{4} q} 
    + \frac{1}{2} q,
\end{equation*}
the Okamoto sigma form for the second \Painleve equation in \cite{okam3} coincides with our sigma form.
\end{rem}
\end{exmp}

\begin{exmp}
Set $n = 2$ in \eqref{SF1}:
\begin{gather*}
    - f_2^{\prime} 
    + \brackets{f_2^{\prime}}^2 
    + \brackets{z - 2 f_2}
    \brackets{
    t_1 \left(\widehat{\mathcal L}_1''
    - \widehat{\mathcal L}_{2}\right)
    + \left(\widehat{\mathcal L}_2''
    - \widehat{\mathcal L}_{3}\right)
    + 2 f_1 f_2 + \sigma_2 \brackets{z}
    - \frac{1}{2} t_{1} z
    }
    = \alpha_2 \brackets{\alpha_2 - 1},
\end{gather*}
where
\begin{gather*}
    f_1
    = \widehat{\mathcal{L}}_1 \LieBrackets{\sigma_2^{\prime} - \dfrac{t_1}{2}},
    \quad 
    f_2
    = t_1 \widehat{\mathcal{L}}_1 \LieBrackets{\sigma_2^{\prime} - \dfrac{t_1}{2}}
    + \widehat{\mathcal{L}}_2 \LieBrackets{\sigma_2^{\prime} - \dfrac{t_1}{2}},
\end{gather*}
with the Lenard operators
\begin{gather*}
    \widehat{\mathcal{L}}_1 
    \LieBrackets{\sigma_2^{\prime} - \dfrac{t_1}{2}}
    = \sigma_2^{\prime} - \dfrac{t_1}{2},
    \qquad
    \widehat{\mathcal{L}}_2 
    \LieBrackets{\sigma_2^{\prime} - \dfrac{t_1}{2}}
    = \sigma_2^{\prime\prime\prime}
    + 3 \brackets{\sigma_2^{\prime}}^2 
    - 3 t_1 \brackets{
    \sigma_2^{\prime}
    - \dfrac{1}{4} t_1},
\end{gather*}
\begin{align*}
    \widehat{\mathcal{L}}_3 
    \LieBrackets{\sigma_2^{\prime} - \dfrac{t_1}{2}}
    = \sigma_2^{(\text{v})}
    &+ 10 \sigma_2^{\prime} \sigma_2^{\prime\prime\prime}
    + 5 \brackets{\sigma_2^{\prime\prime}}^2 
    + 10 \brackets{\sigma_2^{\prime}}^3 
    - 5 t_1 \brackets{\sigma_2^{\prime\prime\prime}
    + 3 \brackets{\sigma_2^{\prime}}^2}
    + \dfrac{5}{2} t_1^2 \brackets{
    3 \sigma_2^{\prime} 
    - \dfrac{1}{2} t_1}.
\end{align*}

Then the sigma form for $\PIIn{2}$ is
\begin{align*}
    &
    \brackets{\sigma_2^{(\text{iv})}}^2 
    + 12 \sigma_2^{\prime} \sigma_2^{\prime\prime} \sigma_2^{(\text{iv})}
    - 4 t_1 \sigma_2^{\prime\prime} \sigma_2^{(\text{iv})}
    - \sigma_2^{(\text{iv})}
    + 4 \sigma_2^{\prime} \sigma_2^{\prime\prime\prime}
    - 2 t_1 \brackets{\sigma_2^{\prime\prime\prime}}^2
    \\
    & \quad
    - 2 \brackets{\sigma_2^{\prime\prime}}^2 \sigma_2^{\prime\prime\prime}
    + 20 \brackets{\sigma_2^{\prime}}^3 \sigma_2^{\prime\prime\prime}
    - 24 t_1 \brackets{\sigma_2^{\prime}}^2 \sigma_2^{\prime\prime\prime}
    + 9 t_1^2 \sigma_2^{\prime} \sigma_2^{\prime\prime\prime}
    - 2 z \sigma_2^{\prime} \sigma_2^{\prime\prime\prime}
    \\
    & \quad \quad
    - t_1^3 \sigma_2^{\prime\prime\prime}
    + 2 z t_1 \sigma_2^{\prime\prime\prime}
    - 2 \sigma_2 \sigma_2^{\prime\prime\prime}
    + 30 \brackets{\sigma_2^{\prime}}^2 \brackets{\sigma_2^{\prime\prime}}^2 
    - 20 t_1 \sigma_2^{\prime} \brackets{\sigma_2^{\prime\prime}}^2 
    + \dfrac{7}{2} t_1^2 \brackets{\sigma_2^{\prime\prime}}^2
    \\
    & \quad \quad \quad
    + z \brackets{\sigma_2^{\prime\prime}}^2
    - 6 \sigma_2^{\prime} \sigma_2^{\prime\prime}
    + 2 t_1 \sigma_2^{\prime\prime}
    + 24 \brackets{\sigma_2^{\prime}}^5
    - 46 t_1 \brackets{\sigma_2^{\prime}}^4
    + 34 t_1^2 \brackets{\sigma_2^{\prime}}^3
    - 4 z \brackets{\sigma_2^{\prime}}^3
    \\
    & \quad \quad \quad \quad
    - 12 t_1^3 \brackets{\sigma_2^{\prime}}^2 
    + 8 z t_1 \brackets{\sigma_2^{\prime}}^2 
    - 6 \sigma_2 \brackets{\sigma_2^{\prime}}^2
    + 2 t_1^4 \sigma_2^{\prime}
    - 4 z t_1^2 \sigma_2^{\prime}
    + 4 t_1 \sigma_2 \sigma_2^{\prime}
    \\
    & \quad \quad \quad \quad \quad
    - \dfrac{1}{8} t_1^5 
    + \dfrac{1}{2} t_1^3 z
    - \dfrac{1}{2} t_1^2 \sigma_2
    - \dfrac{1}{2} t_1 z^2
    + z \sigma_2
    = \alpha_2 \brackets{\alpha_2 - 1}.
\end{align*}
\end{exmp}

\section{\Backlund transformations}

The \Backlund transformations of \eqref{PIIhier} have two generators  \cite{okam3} 
\begin{align*}
    {s}:
    &
    \qquad 
     \Tilde{w}_n \brackets{z, {t}; \tilde{\alpha}_n}
        = w_n \brackets{z, {t}; \alpha_n} 
        - \dfrac{2 \alpha_n - 1}{2 \dsum_{l = 0}^{n} t_l \mathcal{L}_l \LieBrackets{w_n^{\prime} - w_n^2}},
    \quad
    \tilde{\alpha}_n
    = 1 - \alpha_n
    \\
    {r}:
    &
    \qquad 
     {w_n} \brackets{z, {t}; -{\alpha}_n}
        = - w_n\brackets{z, {t}; {\alpha}_n}.
\end{align*}

\begin{thm}
The \Backlund transformations of the sigma forms \eqref{SF1}, \eqref{SF1n1} act on the sigma function 
\begin{align*}
    {s} \brackets{\sigma_n}
    &= \sigma_n,
    &
    {r} \brackets{\sigma_n}
    &= \sigma_n-2 w_n.
\end{align*}
\end{thm}
\begin{proof}
Since the right hand side of the sigma form \eqref{SF1}, \eqref{SF1n1} is invariant under the $s$-action, $\sigma_n$ is also invariant under this action. 
    Regarding the $r$-action, we have
    \begin{equation*}
        {r} \brackets{w_n^{\prime} - w^2}
        = - w_n^{\prime} - w_n^2
        = \sigma_n^{\prime} - 2 w_n^{\prime}.
    \end{equation*}
    After integration w.r.t. $z$ we obtain the final formula.

\end{proof}

\bibliographystyle{plain}
\bibliography{bibliography}
\nocite{*}

\end{document}